\documentclass[a4paper,USenglish,autoref,thm-restate,cleveref]{lipics-v2021}
\nolinenumbers
\usepackage{amsmath,amssymb,amsthm}
\usepackage{amsfonts}
\usepackage{booktabs}
\usepackage{hyperref}
\usepackage{todonotes}
\usepackage{xcolor}
\usepackage{xspace}
\usepackage[algoruled,vlined]{algorithm2e}
\DontPrintSemicolon
\SetArgSty{}
\SetKw{KwOr}{or}
\SetKw{KwAnd}{and}
\SetKw{KwNot}{not}

\newcommand{\cB}{\ensuremath{\mathcal{B}}\xspace}
\newcommand{\cC}{\ensuremath{\mathcal{C}}\xspace}

\newcommand{\cI}{\ensuremath{\mathcal{I}}\xspace}

\newcommand{\cT}{\ensuremath{\mathcal{T}}\xspace}
\newcommand{\NN}{\ensuremath{\mathbb{N}}\xspace}

\definecolor{defblue}{rgb}{0.121,0.47,0.705}
\let\emph\relax\DeclareTextFontCommand{\emph}{\color{defblue}\em}

\graphicspath{{figures/}}

\newcommand{\PQ}{PQ}
\newcommand{\PQP}{P}
\newcommand{\PQQ}{Q}
\newcommand{\MPQ}{MPQ}

\newcommand{\ie}{i.e.}

\title{Coloring and Recognizing Mixed Interval Graphs}

\author{Grzegorz Gutowski}{Theoretical Computer Science Department, Faculty of Mathematics and Computer Science, Jagiellonian University, Krak\'ow, Poland}{grzegorz.gutowski@uj.edu.pl}{https://orcid.org/0000-0003-3313-1237}{partially supported by the National Science Center of Poland under grant no.\ 2019/35/B/ST6/02472.}
\author{Konstanty Junosza-Szaniawski}{Warsaw University of Technology, Warsaw, Poland}{}{0000-0003-0352-8583}{}
\author{Felix Klesen}{Universit\"at W\"urzburg, W\"urzburg, Germany}{}{0000-0003-1136-5673}{}
\author{Pawe\l~Rz\k{a}\.zewski}{Warsaw University of Technology, Warsaw, Poland \and Institute of Informatics, University of Warsaw, Warsaw, Poland}{}{0000-0001-7696-3848}{}
\author{Alexander~Wolff}{Universit\"at W\"urzburg, W\"urzburg, Germany}{}{0000-0001-5872-718X}{}
\author{Johannes~Zink}{Universit\"at W\"urzburg, W\"urzburg,
  Germany}{}{0000-0002-7398-718X}{partially supported by DFG grant Wo 758/11-1.}
\authorrunning{G.~Gutowski, K.~Junosza-Szaniawski, F.~Klesen, P.~Rz\k{a}\.zewski, A. Wolff, and J.~Zink}


\hideLIPIcs\def\subjclassHeading{}\ccsdesc{}\global\renewcommand\ccsdesc[2][100]{}

\begin{document}

\maketitle

\begin{abstract}
  A \emph{mixed interval graph} is an interval graph that has, for
  every pair of intersecting intervals, either an arc (directed
  arbitrarily) or an (undirected) edge.  We are particularly interested in 
  scenarios where edges and arcs are defined by the geometry of intervals.
  In a proper coloring of a mixed interval
  graph $G$, an interval~$u$ receives a lower (different) color than
  an interval~$v$ if $G$ contains arc $(u,v)$ (edge $\{u,v\}$).
  Coloring of mixed graphs has applications, for example, in scheduling with
  precedence constraints; see a survey by Sotskov [Mathematics, 2020].

  For coloring general mixed interval graphs, we present a
  $\min \{\omega(G), \lambda(G)+1 \}$-appro\-xi\-ma\-tion algorithm,
  where $\omega(G)$ is the size of a largest clique and $\lambda(G)$
  is the length of a longest directed path in~$G$.
  For the subclass of \emph{bidirectional interval graphs} (introduced
  recently for an application in graph drawing), we show that optimal
  coloring is NP-hard.  This was known for general mixed interval
  graphs.
  
  We introduce a new natural class of mixed interval graphs, which we
  call \emph{containment interval graphs}.  In such a graph, there is
  an arc~$(u,v)$ if interval~$u$ contains interval~$v$,
  and there is an edge $\{u,v\}$ if $u$ and $v$ overlap.  We show that
  these graphs can
  be recognized in polynomial time, that coloring them
  with the minimum number of colors is NP-hard, and that there is a
  2-approximation algorithm for coloring.
\end{abstract}

\section{Introduction}

In a geometric intersection graph, the vertices represent geometric
objects, and two vertices are adjacent if and only if the
corresponding objects intersect.  For example, \emph{interval graphs}
are the intersection graphs of intervals on the real line.  These
graphs are well understood: interval graphs are \emph{chordal} and can
thus be colored optimally (that is, with the least number of colors)
in polynomial time.  In other words, given an interval graph~$G$, its
\emph{chromatic number} $\chi(G)$ can be computed efficiently.

The notion of coloring can be adapted to directed graphs where an arc
$(u,v)$ means that the color of~$u$ must be smaller than that
of~$v$.  Clearly, such a coloring can only exist if the given graph is
acyclic.  Given a directed acyclic graph, its chromatic number can be
computed efficiently (via topological sorting).

A generalization of both undirected and directed graphs are
\emph{mixed graphs} that have edges and arcs.  A \emph{proper
  coloring} of a mixed graph~$G$ with vertex set~$V(G)$ is a function
$f \colon V(G) \to \NN$ such that, for any distinct vertices~$u$
and~$v$ of~$G$, the following %
conditions hold:
\begin{enumerate}
\item if there is an edge $\{u,v\}$, then $f(u) \neq f(v)$, and
\item if there is an arc $(u,v)$, then $f(u) < f(v)$.
\end{enumerate}
The objective is to minimize the number of colors.

The concept of mixed graphs was introduced by Sotskov and
Tanaev~\cite{SotskovT76} and reintroduced by Hansen, Kuplinsky, and de
Werra~\cite{HansenKW97} in the context of proper colorings of mixed
graphs.
Properly coloring mixed graphs is NP-hard even for bipartite
planar graphs~\cite{RiesW08} but admits efficient algorithms for
trees~\cite{FurmanczykKZ08Trees} and series-parallel
graphs~\cite{FurmanczykKZ08SP}.

For a mixed interval graph~$G$, the \emph{underlying undirected graph}
of~$G$, denoted by $U(G)$, has an edge for every edge or arc of~$G$.
Note that testing whether a given graph~$G$ is a mixed interval graph
means testing whether $U(G)$ is an interval graph, which
takes linear time~\cite{LuekerB79}.

\subparagraph*{Motivation.}

Mixed graphs are graphs where some vertices are connected by
(undirected) edges and others by directed arcs.
Such structures are useful for modeling relationships
that involve both directed and undirected connections and find
applications in various areas, including network analysis,
transportation planing, job scheduling, and circuit design.

Coloring mixed graphs is relevant in task scheduling problems
where tasks have dependencies and resource requirements 
\cite{TanaevSS94,Brucker95,SotskovTW02}.
In circuit design, coloring mixed graphs allows us to reduce signal
crosstalk or interference.  Other applications include modeling of
metabolic pathways in biology~\cite{BeckBCJY-SNC12},
process management in operating systems~\cite{BeckBCJY-GC15},
traffic signal synchronization~\cite{SerafiniU1989}, and
timetabling~\cite{deWerra-DM97}.  See the extensive survey by
Sotskov~\cite{Sotskov20} for other applications and relevant problems.
Coloring of mixed graphs is a challenging problem, as many techniques
known for solving graph coloring problems fail in the more general setting.

The study of mixed graph coloring for interval graphs was initiated by
Zink et al.~\cite{zwbw-ldugg-CGTA22}, motivated by the minimization of
the number of additional sub-layers for routing edges in layered orthogonal graph drawing
according to the so-called Sugiyama framework~\cite{stt-mvuhss-TSMC81}.
In a follow-up paper, Gutowski et al.~\cite{gmrswz-cmdig-GD22} resolved some of the
problems concerning interval graphs where the subset of arcs and their
orientations are given by the geometry of the intersecting intervals.
Driven by the graph drawing application, they focused on the
\emph{directional variant} where, for every pair of intersecting intervals,
there is an edge when one interval is contained in the other and there
is an arc oriented towards the right interval when the intervals overlap.

In this paper we focus on the \emph{containment variant} where, for any pair
of intersecting intervals, there is an arc oriented towards the
smaller interval when one interval is contained in the other,
and an edge when they overlap.
This is the only other natural geometric variant that can be defined
for interval graphs, but the containment variant can also be defined
for other geometric intersection graphs, or even for graphs defined by
systems of intersecting sets.

As there are already some effective techniques for the directional
variant, our hope was to use them in the containment variant, or even in more general settings.
Quite unexpectedly, the containment variant for interval
graphs turned out to be more difficult than the directional variant.
We have found our techniques for proving lower bounds for the containment variant of
interval graphs to be applicable for the \emph{bidirectional variant}
considered previously~\cite{zwbw-ldugg-CGTA22,gmrswz-cmdig-GD22}.
This variant is a generalization of the directional variant mentioned
above.  Every interval has an orientation; left-going or right-going.
There is an arc between two intervals if and only if they overlap
and their orientations agree.  Arcs between left-going intervals are
directed as in the directional variant; the condition for right-going
intervals is symmetric.
As a result, we get that minimizing the number of additional
sub-layers in layered orthogonal graph drawing according to the
Sugiyama framework is NP-hard.

\subparagraph*{Our Contribution.}

In this paper we forward the study of coloring mixed graphs where edge
directions have a geometric meaning.  To this end, we
introduce a new natural class of mixed interval graphs, which we
call \emph{containment interval graphs}.  In such a graph, there is an
arc~$(u,v)$ if interval~$u$ contains interval~$v$, and there is an
edge $\{u,v\}$ if $u$ and $v$ overlap.  For a set~\cI of intervals,
let $\cC[\cI]$ be the containment  
interval graph induced by~\cI.  We show that these graphs can be
recognized in polynomial time (\cref{sec:short-recognition}), that
coloring them optimally is NP-hard
(\cref{sec:coloring-containment-hard}), and that, for every set~\cI of
intervals, it holds that $\chi(\cC[\cI]) \le 2\omega(\cC[\cI])-1$, that is,
$\cC[\cI]$ can be colored with fewer than twice as many colors as the size
of the largest clique in~$\cC[\cI]$
(\cref{sec:two-approx-coloring-containment}).  In other words,
containment interval graphs are \emph{$\chi$-bounded}.
Our constructive proof yields a 2-approximation algorithm for coloring
containment interval graphs.

Then we prove that, for the class of bidirectional interval graphs,
optimal coloring is NP-hard (\cref{sec:bidirectional-NP-hard}).
This answers (negatively) an open problem that was asked
previously~\cite{gmrswz-cmdig-GD22}.  Our reduction is similar to the
one for containment interval graphs, but technically somewhat more
challenging.  Finally, we show that, for any mixed interval graph~$G$
without directed cycles, it holds that
$\chi(G) \leq \omega(G) \cdot (\lambda(G)+1)$,
where $\lambda(G)$ denotes the length of a longest
directed path in~$G$ (\cref{sec:coloring-general}).  Since
$\chi(G) \ge \max \{\omega(G), \lambda(G)+1\}$,
our constructive proof for the upper bound yields a
$\min \{\omega(G), \lambda(G)+1 \}$-approximation algorithm.
The upper bound is asymptotically tight in the worst case.

\begin{table}[tb]
  \centering
  \caption{Known and new results concerning subclasses of mixed
    interval graphs.  The time complexities refer to a given set of
    $n$ intervals with $m$ pairwise intersections. (We use
    {\color{gray}T}, {\color{gray}P}, and {\color{gray}C}
    as shorthand for Theorem, Proposition, Corollary, respectively.)}
  \label{tab:results}

  \smallskip\small

  \begin{tabular}{@{}l@{~~}c@{~~}>{\color{gray}}c@{\quad}c@{~~}>{\color{gray}}c@{\quad}c@{~~}>{\color{gray}}c@{\quad}c@{~~}>{\color{gray}}c@{\quad}c@{~~}>{\color{gray}}c@{}}
    \toprule
    Mixed interval & \multicolumn{8}{c@{\quad}}{Coloring}
                   & \multicolumn{2}{@{}c@{}}{Recognition} \\
    graph class    & \multicolumn{2}{@{}c@{\quad}}{complexity}
                   & \multicolumn{2}{@{}c@{\quad}}{lower bound}
                   & \multicolumn{2}{@{}c@{\quad}}{upper bound}
                   & \multicolumn{2}{@{}c@{\quad}}{approximation} & \\
    \midrule
    containment    & NP-hard & T\ref{thm:containment-NP-hard}
                   & $2\omega{-}1$ & P\ref{prop:containment-two-omega}
                   & $2\omega{-}1$ & T\ref{thm:containment-two-omega}
                   & 2 & C\ref{cor:containment-two-approx}
                   & $O(nm)$ & T\ref{thm:recognition} \\
    directional    & $O(n \log n)$ & \cite{gmrswz-cmdig-GD22}
                   & & & & & 1 & \cite{gmrswz-cmdig-GD22}
                   & $O(n^2)$ & \cite{gmrswz-cmdig-GD22} \\
    bidirectional  & NP-hard & T\ref{thm:bidirectional-NP-hard}
                   & & & & & 2 & \cite{gmrswz-cmdig-GD22}
                   & \multicolumn{2}{c@{}}{open} \\
    general        & NP-hard & \cite{gmrswz-cmdig-GD22}
                   & $(\lambda{+}2)\omega/2$ & P\ref{prop:mixed-lower-bound}
                   & $(\lambda{+}1)\omega$ & T\ref{thm:mixed-upper-bound}
                   & $\min \{\omega, \lambda{+}1 \}$ 
                   & T\ref{thm:mixed-upper-bound}
                   & $O(n{+}m)$ & \cite{LuekerB79} \\
    \bottomrule
  \end{tabular}
\end{table}

\cref{tab:results} gives an overview over known and new results
concerning the above-mentioned subclasses of mixed interval graphs.
Given a positive integer $k$, we use $[k]$ as shorthand for the set
$\{1, 2, \dots, k\}$.
When we visualize a graph coloring corresponding to a set of intervals,
we use horizontal tracks to indicate the color.
In \cref{fig:graph-classes}, we briefly analyze the relationships
between the three classes of geometrically defined mixed interval
graphs; directional~($\mathcal D$), bidirectional~($\mathcal B$), and
containment interval graphs~($\mathcal C$).

\begin{figure}[t]

  \null\hfill
  \begin{subfigure}[b]{.25\textwidth}
    \centering
    \includegraphics[page=2]{graph-classes}
    \caption{$G_1 \in \mathcal{C} \setminus \mathcal{B}$}
  \end{subfigure}
  \hfill
  \begin{subfigure}[b]{.25\textwidth}
    \centering
    \includegraphics[page=3]{graph-classes}
    \caption{$G_2 \in \mathcal{D} \setminus \mathcal{C}$}
  \end{subfigure}
  \hfill
  \begin{subfigure}[b]{.25\textwidth}
    \centering
    \includegraphics[page=1]{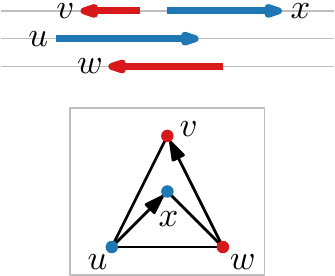}
    \caption{$G_3 \in \mathcal{B} \setminus \mathcal{D}$}
  \end{subfigure}
  \hfill\null

  \caption{Let $\mathcal D$, $\mathcal B$, and $\mathcal C$ be the
    classes of directional, bidirectional, and containment interval
    graphs.  Clearly, $\mathcal{D} \subseteq \mathcal B$.  The above
    sets of intervals and the corresponding directed graphs show that
    the classes~$\mathcal D$ and~$\mathcal B$ are incomparable with
    the class~$\mathcal{C}$ and that $\mathcal{D}$ is properly
    contained in~$\mathcal B$.}
  \label{fig:graph-classes}
\end{figure}

\section{Recognition of Containment Interval Graphs}
\label{sec:short-recognition}

Booth and Lueker~\cite{LuekerB79} 
introduced a data structure called \emph{\PQ-tree}
to recognize, for a given undirected graph~$G$, whether $G$ is an interval graph.
A \PQ-tree is a rooted tree of so-called \emph{\PQP-nodes},
where the order of children can be arbitrarily permuted,
and \emph{\PQQ-nodes}, where the order of children is fixed up to inversion.
A specific permutation of all nodes is called a \emph{rotation}.
One can think of the leaves of a \PQ-tree to represent the maximal cliques of~$G$
and a specific rotation to represent an order of the maximal cliques,
which implies an interval representation of~$G$
where every vertex is contained in a consecutive sequence of maximal cliques.
(Actually, a \PQ-tree can encode all possible interval representations of~$G$.)

Observe that a %
representation %
of a containment interval graph is an interval representation.
Hence, if, for a given mixed graph~$G$, a containment
representation~\cI exists, then \cI corresponds to a rotation of the
\PQ-tree constructed for the underlying undirected graph~$U(G)$
by the algorithm of Booth and Lueker~\cite{LuekerB79}.
Hence, to recognize a containment interval graph~$G$,
we proceed in three phases.
First, we compute a \PQ-tree $T$ of $U(G)$.
Second, we find a rotation of $T$ corresponding
to a containment representation of~$G$.
Third, we determine suitable endpoints for the intervals
corresponding to our selected rotation resulting in
a containment representation $\cI$ of $G$.

In the second phase,
we proceed top-down to fix the permutation of each node of~$T$
while we maintain as invariant that
before and after deciding the permutation of a single node,
we can still reach a rotation of~$T$
corresponding to a containment representation $\cI$
(provided $G$ is a containment interval graph).
Depending on the set of maximal cliques (corresponding to leaves)
a vertex~$v$ is contained in,
we can determine where $v$ is \emph{introduced} in~$T$
(roughly at the root of the subtree containing all leaves corresponding to~$v$).
Intuitively, it is ``natural'' for a vertex $u$ introduced further up in the
tree to have an arc towards a vertex~$v$ introduced further down in the tree.
However, if $u$ and $v$ are connected by an edge,
we need to permute the nodes of the \PQ-tree such that
both $u$ and $v$ start or end in the same maximal clique.
These restrictions can propagate.

If we end up with a rotation of~$T$,
we construct, in the third phase, a corresponding containment representation if possible.
To this end, we determine for every
vertex the first and the last clique it appears in,
which groups the left and right endpoint of the intervals.
Within each group,
we sort the endpoints according to the constraints
implied by the arcs and edges where possible.
What remains are induced mixed subgraphs
of vertices that start and end in the same cliques
and that behave the same with respect to every other vertex
(i.e., they are all connected to this vertex by an outgoing arc or an incoming arc or an edge).
We can interpret each such subgraph as a \emph{partially ordered set}
for which we need to check whether it is \emph{two-dimensional}
and find two corresponding linear orders,
which gives us an ordering of their left and their right endpoints.
This last part depends on the linear-time algorithm by McConnell and Spinrad~\cite{McConnellS99} that can construct such two orders for any two-dimensional poset.

\begin{restatable}[{\hyperref[thm:recognition*]{$\star$}}]{theorem}{ThmRecognition}
        \label{thm:recognition}
	There is an algorithm that, given a mixed graph~$G$, decides whether~$G$ is a containment interval graph.
	The algorithm runs in $O(nm)$ time,
	where $n$ is the number of vertices of~$G$
	and $m$ is the total number of edges and arcs of~$G$,
	and produces a containment representation of~$G$ if~$G$ admits one.
\end{restatable}

The full proof follows the ideas presented above,
but has some technical subtleties; see the appendix.

\section{A 2-Approximation Algorithm for Coloring Containment Interval
  Graphs}
\label{sec:two-approx-coloring-containment}

In this section, we present a 2-approximation algorithm for coloring
containment interval graphs, we detail how to make the algorithm run
in $O(n \log n)$ time for a set of $n$ intervals, and we construct a
family of sets of intervals that shows that our analysis is tight.

\begin{theorem}
  \label{thm:containment-two-omega}
  For any set $\cI$ of intervals, the containment interval graph
  $\cC[\cI]$ induced by \cI admits a proper
  coloring with at most $2 \cdot \omega(\cC[\cI]) - 1$ colors.
\end{theorem}
\begin{proof}
For simplicity, let $G := \cC[\cI]$ and $\omega := \omega(G)$.
We use induction on $\omega$. If $\omega=1$, then $G$ has no edges and clearly admits a proper coloring using only one color.
So assume that $\omega>1$ and that the theorem holds for all graphs
with smaller clique number.

Recall that a \emph{proper} interval graph is an interval graph that
has a representation where no interval is contained in another interval.
Let $M(\cI)$ denote the subset of $\cI$ consisting of intervals that
are maximal with respect to the containment relation.
In particular, $\cC[M(\cI)]$ is a proper interval graph.  Observe that
$\bigcup M(\cI) = \bigcup \cI$ (where we consider the union of
intervals as a subset of the real line).  Let $R$ be an inclusion-wise
minimal subset of $M(\cI)$ such that $\bigcup R = \bigcup \cI$.
In \cref{fig:example}, the intervals in $M(\cI)$ are marked with
crosses and the set of intervals on the lowest two (gray) lines is one
way of choosing~$R$.

\begin{claim}
  $\cC[R]$ is an undirected linear forest.
\end{claim}
\begin{claimproof}
  All intervals in $M(\cI)$ and thus in $R$ are incomparable with
  respect to the containment relation, so $\cC[R]$ has no arcs.  Note
  that $\cC[R]$ is a proper interval graph, so it contains no
  induced $K_{1,3}$ and no induced cycle with at least four
  vertices.  Thus it suffices to prove that $\cC[R]$ is triangle-free.
  For contradiction, suppose otherwise.  Let $x,y,z$ induce a triangle
  in $\cC[R]$, ordered according to their left endpoints.  As $x,y,z$
  are pairwise overlapping, note that $y \subseteq x \cup z$,
  and thus $\bigcup (R \setminus \{y\}) = \bigcup R$. This contradicts
  the minimality of~$R$.
\end{claimproof}

By the claim above, $\cC[R]$ can be properly colored with colors
$\{1,2\}$.  Let $f_1$ be such a coloring.
If $R = \cI$, we are done (using only $\omega$ many colors), so
suppose that $\cI \setminus R \ne \emptyset$.
Slightly abusing notation, we define $G' := G - R$.

\begin{claim}
The largest clique in $G'$ has at most $\omega-1$ vertices.
\end{claim}
\begin{claimproof}
  As $G'$ is a subgraph of $G$, each clique in $G'$ has at most
  $\omega$ vertices.  For contradiction, suppose that there is a set
  $S \subseteq \cI \setminus R$ such that $|S|=\omega$ and all
  intervals in $S$ pairwise intersect.  By the Helly property of
  intervals, $\bigcap S \ne \emptyset$.
  Let $p \in \bigcap S$.  Since
  $\bigcup R = \bigcup \cI$, there is an interval $r \in R$ that
  contains~$p$.  Thus $S \cup \{r\}$ is a clique in $G$ with
  $\omega +1$ vertices, which contradicts the definition of $\omega$.
\end{claimproof}

By the inductive assumption, $G'$ admits a proper coloring $f_2$ using colors $[2(\omega-1)-1]$.
Finally, we define $f \colon \cI \to [2\omega-1]$ as follows:
\[
f(x) = \begin{cases}
f_1(x) & \text{ if } x \in R,\\
f_2(x)+2 & \text{ if } x \notin R.
\end{cases}
\]
We claim that $f$ is a proper coloring of $G$.  (For an example, see
\cref{fig:example}.)

\begin{figure}[tb]
  \centering
  \includegraphics{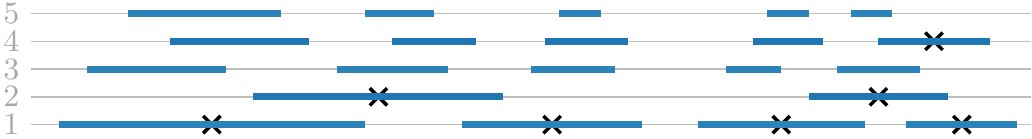}
  \caption{A set of intervals and a coloring produced by the
    2-approximation algorithm.  The intervals that lie in $M(\cI)$ at
    the top level of the recursion are marked with crosses.}
  \label{fig:example}
\end{figure}

First, note that if $x,y \in \cI$ are distinct and $x \cap y \neq \emptyset$, then $f(x) \neq f(y)$.
Indeed, if $x,y \in R$, then $f(x) = f_1(x) \neq f_1(y)=f(y)$.
If $x,y \notin R$, then $f(x) = f_2(x) +2 \neq f_2(y) + 2=f(y)$.
Finally, if, say, $x \in R$ and $y \notin R$, then $f(x) \in \{1,2\}$ and $f(y) \geq 3$.

It remains to argue that the second condition in the definition of a
proper coloring holds as well.
For a contradiction, let~$x$ and~$y$ be distinct intervals and assume
that $x \subseteq y$ and $f(y) > f(x)$.
Note that $x \notin M(\cI)$ and thus $x \notin R$.
This implies that $f(x)\ge3$.
Since we assumed that $f(y)>f(x)$, we have that $f(y)>3$.
Hence, $y \not\in R$.  However, by the inductive assumption,
we have that $f(x) = f_2(x) +2 > f_2(y)+2 = f(y)$, which yields
the desired contradiction.  This completes the proof.
\end{proof}

Observe that the proof of \cref{thm:containment-two-omega} can be
easily transformed into an efficient algorithm, which yields the
following corollary.

\begin{corollary}
  \label{cor:containment-two-approx}
  There is a 2-approximation algorithm for coloring interval
  containment graphs properly.  Given a set of $n$ intervals, the
  algorithm runs in $O(n \log n)$ time.  
\end{corollary}

\begin{proof}
  For any graph~$G$, we have $ \chi(G) \geq \omega(G)$.  Hence, the
  approximation factor follows directly from
  \cref{thm:containment-two-omega}.
  
  It remains to implement the constructive proof of
  \cref{thm:containment-two-omega} efficiently.
  Let~\cI be the given set of intervals.  For each
  interval~$I$ in \cI, let~$r_I$ be the right endpoint of~$I$.
  We go through the intervals from left to right in several phases.
  In each phase, we use two colors, except possibly in the last phase
  where we may use only one color.  For phase~$i$ with $i \ge 1$, we
  reserve the set $\mathit{colors}(i)=\{2i-1,2i\}$.
  We use an augmented balanced binary search tree~\cT to store the
  intervals in~\cI.  We will query~\cT in two ways.  A query of
  type~Q1 in~\cT with a
  value~$x \in \mathbb{R} \cup \{-\infty\}$ will return, among all
  intervals whose left endpoint is at least~$x$, one with leftmost
  left endpoint (and $\mathit{nil}$ if such an interval does not exist).
  A query of type Q2 in \cT with a value $y \in \mathbb{R}$ will return,
  among all intervals whose left endpoint is at most~$y$, one with
  rightmost right endpoint (and $\mathit{nil}$ if such an interval
  does not exist).  Note that the two queries are not symmetric.

  Algorithm~\ref{alg:2-approx}
  describes our algorithm in pseudocode.
  Initially, \cT stores all intervals in~\cI.  The algorithm
  terminates once \cT is empty and all intervals are colored.

  \newcommand\mygraycomm[1]{\ttfamily\textcolor{lightgray}{#1}}
  \newcommand\myblackcomm[1]{\ttfamily\textcolor{black}{#1}}
  \newcommand{\mygraytcp}[1]{\SetCommentSty{mygraycomm}\tcp*{#1}\SetCommentSty{myblackcomm}}
  \begin{algorithm}
    \caption{\textsc{2-Approximate-Coloring}(set \cI of $n$ intervals)}
    \label{alg:2-approx}
    $\cT.\mathsf{initialize}(\cI)$ \;
    $i=0$ \;
   \While{\KwNot $\cT.\mathsf{empty}()$}{
      $i = i+1$ \tcp*{start new phase}
      $I=\cT.\mathsf{Q1}(-\infty)$\;
      $c = I.\mathit{color} = 2i-1$
      \mygraytcp{color $I$ and set current color}%
      $\cT.\mathsf{remove}(I)$ \;
      \While{\KwNot $\cT.\mathsf{empty}()$}{
        $I' = \cT.\mathsf{Q2}(r_I)$ \;
        \If(\tcp*[f]{no interval contains $r_I$})%
          {$I'==\mathit{nil}$ \KwOr $r_{I'} \le r_I$}{
          $I'= \cT.\mathsf{Q1}(r_I)$ \;
          \lIf{$I'==\mathit{nil}$}{%
            \textbf{break}
            \tcp*[f]{finish current phase}%
          }
          $I'.\mathit{color} = c$
          \mygraytcp{color $I'$ and keep current color}
        }
        \Else(\tcp*[f]{$I$ and $I'$ overlap}){
          $c = I'.\mathit{color} = \mathit{colors}(i) \setminus \{c\}$
          \mygraytcp{color $I'$ and swap current color}
        }
        $\cT.\mathsf{remove}(I')$ \;
        $I=I'$
      }
    }
  \end{algorithm}

  We start each phase by Q1-querying~\cT with $-\infty$.  This yields
  the leftmost interval~$I$ stored in~\cT.  We color~$I$ with the
  smaller color~$2i-1$ reserved for the current phase~$i$.  Let the
  \emph{current color} $c$ be this color.  We remove~$I$ from~\cT.
  Then we Q2-query~\cT with the right endpoint~$r_I$ of~$I$ and
  consider the following two possibilities.
  \begin{description}
  \item[Case~I:] If the Q2-query returns $\mathit{nil}$ or an interval that
  lies completely to the left of~$r_I$, we Q1-query~\cT with $r_I$ for
  an interval to the right of~$r_I$.  If such an interval~$I'$ exists,
  it must be disjoint from~$I$, so we color~$I'$ with the current
  color~$c$.  Otherwise, we start a new phase.
  \item[Case~II:] If the Q2-query returns an interval $I'$ that overlaps with
  the previous interval~$I$, we color~$I'$ with the other color $c'$
  that we reserved for the current phase, that is,
  $\{c'\}=\mathit{colors}(i) \setminus \{c\}$.  Then we set the
  current color~$c$ to~$c'$.
  \end{description}
  In either case, if we do not start a new phase, we remove $I'$
  from~\cT and proceed with the next Q2-query as above, with $I'$ now
  playing the role of~$I$.

  It remains to implement the balanced binary search tree~\cT.  
  The key of an interval is its left endpoint.  For
  simplicity, we assume that the intervals are stored in the leaves
  of~\cT and that the key of each inner node is the maximum of the
  keys in its left subtree.  This suffices to answer queries of type~Q1. 
  For queries of type Q2, we augment~\cT by storing, with each
  node~$\nu$, a value $\max(\nu)$ that we set to the maximum of the
  right endpoints among all intervals in the subtree rooted at~$\nu$.
  (We also store a pointer $\mu(\nu)$ to the interval that yields the
  maximum.)  In a Q2-query with a value~$y$, we search for the largest
  key~$k \le y$.  Let~$\pi$ be the search path in~\cT, and initialize~$m$
  with~$-\infty$.  When traversing~$\pi$, we inspect each node~$\nu$ that 
  hangs off~$\pi$ on the left side.  If $\max(\nu)>m$, then we set
  $m=\max(\nu)$ and $\rho=\mu(\nu)$.  When we reach a leaf, $\rho$
  points to an interval whose right endpoint is maximum among
  all intervals whose left endpoint is at most~$y$.

  The runtime of $O(n \log n)$ is obvious since we insert, query, and
  delete each interval in $O(\log n)$ time exactly once.
\end{proof}

\begin{proposition}
  \label{prop:containment-two-omega}
  There is an infinite family $(\cI_n)_{n \ge 1}$ of sets of intervals
  with $|\cI_n|=3 \cdot 2^{n-1}-2$, %
  $\chi(\cC[\cI_n])=2n-1$, and $\omega(\cC[\cI_n])=n$.
\end{proposition}

\begin{proof}
The construction is iterative.
The family $\cI_1$ consists of a single interval of unit length.

Now let $n>1$ and suppose that we have defined~$\cI_{n-1}$ and want to
define~$\cI_n$.
We introduce two new intervals $\ell_n$ and $r_n$, both of length $3^{n-1}$, that overlap slightly.
Then we introduce two copies of $\cI_{n-1}$.
All intervals of one copy are contained in $\ell_n \setminus r_n$, and all intervals of the other copy are contained in $r_n \setminus \ell_n$.

The number of intervals in $\cI_n$ is given by the recursion: $f(1)=1$ and $f(n) = 2f(n-1)+2$, 
which solves to $f(n) = 3 \cdot 2^{n-1} - 2$. Furthermore, it is straightforward to observe that with each step of the construction,
the size of a largest clique increases by $1$.

We claim that, for $i \in [n]$,
in any proper coloring of $\cC[\cI_i]$, the difference between the largest
and the smallest color used is at least $2i-2$. 
Clearly, the claim holds for $i=1$.  Now~assume that it holds for
$i=n-1$.  Consider any proper coloring of~$\cC[\cI_n]$, and let $m$ be
the minimum color used in this coloring.
The colors of $\ell_n$ and $r_n$ must be different.
Without loss of generality, suppose that the color of $r_n$ is larger
than the color of $\ell_n$.
In particular, the color of $r_n$ is at least $m+1$.
Now consider the copy of $\cI_{n-1}$ contained in $r_n$. 
The color of each interval in this copy must be larger than the
color of $r_n$, so in particular the minimum color used for this copy
of $\cI_{n-1}$ is at least $m+2$.
By the inductive assumption, some interval in this copy of~$\cI_{n-1}$
receives a color that is at least $m+2 + 2(n-1)-2 = 2n-2+m$.
Summing up, the difference between the largest and the smallest color
used for $\cC[\cI_n]$ is at least $2n-2$.

Given that the minimum color is~$1$, we conclude that $\chi(\cC[\cI_n]) \geq 2n-1$.

For the upper bound, %
we color~$\cC[\cI_n]$ as follows.
For $n=1$, we color the only interval with color $1$.
For $n > 1$, we color $\ell_n$ with color~$1$ and $r_n$ with color~$2$.
Next, for each of the two copies of $\cI_{n-1}$, we use the proper
coloring defined inductively with all colors increased by $2$, see
\cref{fig:lower-bound-two-omega}.
\end{proof}

\begin{figure}[tb]
  \centering
  \includegraphics{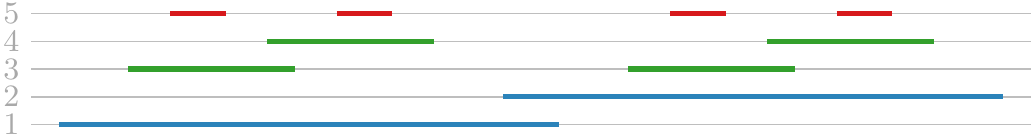}
  \caption{Instance for the proof of \cref{prop:containment-two-omega}
    for $n=3$: $|\cI_n|=3 \cdot 2^{n-1}-2=10$, $\omega(\cI_n)=n=3$, and
    $\chi(\cI_n)=2n-1=5$.}
  \label{fig:lower-bound-two-omega}
\end{figure}

\section{Coloring Containment Interval Graphs Is NP-Hard}
\label{sec:coloring-containment-hard}

In this section we show that it is NP-hard to color a containment
interval graph with a given number of colors.

\begin{theorem}
	\label{thm:containment-NP-hard}
	Given a set \cI of intervals and a positive integer $k$, it is
	NP-hard to decide whether $k$ colors suffice to color $\cC[\cI]$, that
	is, whether $\chi(\cC[\cI]) \le k$.
\end{theorem}

\begin{proof}
	We describe a reduction from (exact) 3-\textsc{Sat}, i.e.,
	the satisfiability problem where every clause contains exactly three literals.
	Let $\varphi = C_1 \wedge C_2 \wedge \dots \wedge C_m$ be an
	instance of 3-\textsc{Sat} where,
	for each clause $C_i$ ($i \in [m]$),
	the literals are negated or unnegated variables from the set
	$\{x_1,x_2,\dots,x_n\}$, and let $H = 5 (m+1)$ be a threshold.

	Using $\varphi$, we construct in polynomial time
	a set of intervals (with pairwise distinct endpoints)
        such that the corresponding containment
	interval graph has a proper coloring with $H$ colors 
	if and only if $\varphi$ is satisfiable.  To this end,
	we introduce \emph{variable gadgets} and \emph{clause gadgets},
	which are sets of intervals representing
	the variables and clauses of $\varphi$, respectively.
	Our main building structure used in these gadgets is
	a \emph{Christmas tree}, %
        that is, an ordered set of intervals where each interval
        contains its successor; see, for example, the set of red
        intervals in \cref{fig:variable-gadget}.  Clearly, the
        intervals of a Christmas tree form a totally ordered clique
        and any proper coloring needs to observe this order.  In
        \cref{fig:clause-gadget}, Christmas trees are represented by
        trapezoids.  The \emph{height} of a Christmas tree is the
        number of intervals it consists of.
	
	Let $j \in [n]$.
	The variable gadget for~$x_j$ consists of two Christmas trees
        (formed by the red and gray intervals in
        \cref{fig:variable-gadget}) whose longest intervals overlap
        and, for each tree, of two additional intervals (green in
        \cref{fig:variable-gadget}).  These green intervals lie
        immediately to the left and to the right of the shortest
        interval in their tree.  The right green interval of the red
        tree overlaps the left green interval of the gray tree.        
	\Cref{fig:variable-gadget} depicts two representations of the
	same gadget for a variable~$x$, each with its own coloring of the
	intervals (encoded by the height of the intervals; see the
	numbers at the right side of the gray box).  The left
	representation with its coloring corresponds to assigning
	true to~$x$, the right representation corresponds to assigning
	false.
	\begin{figure}[tb]
		\centering
		\includegraphics[page=2]{containment-variable-gadget}
		\hfill
		\includegraphics[page=1]{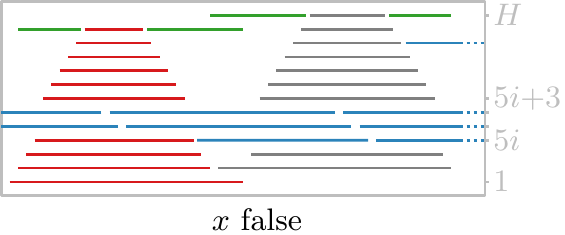}
		
		\caption{Variable gadget for the proof of
			\cref{thm:containment-NP-hard} in its two states.  The blue
			intervals with dots extend to the clause gadgets.  The topmost
			blue interval
			(starting immediately to the right of a gray interval) indicates
			that $x$ is part of a clause~$C_j$ with $j>i$; the blue interval
			(with a small gap) that starts immediately to the right of a red
			interval indicates that $\neg x$ is part of the clause~$C_i$.}
		\label{fig:variable-gadget}
	\end{figure}
	The height of the red tree is $H-1$ minus the number of occurrences
	of literals $x_{j'}$ and $\neg x_{j'}$ with $j'<j$ in~$\varphi$.
	The height of the gray tree is that of the red tree minus the number
	of occurrences of $\neg x_j$ in~$\varphi$.
	We say that $x_j$ is set to true if the bottom interval of
	the gray tree has color~1; otherwise we say that $x_j$ is set
        to false.

	For $i \in [m]$, the gadget for clause $C_i$ consists of a
	Christmas tree (light blue in \cref{fig:clause-gadget}) of height
	$H-(5i+2)=5(m-i)+3$.  All clause gadgets are placed to the
	right of all variable gadgets, in the order $C_1,\dots,C_m$
	from left to right.
        
	\begin{figure}[tb]
		\centering
		\includegraphics[page=1]{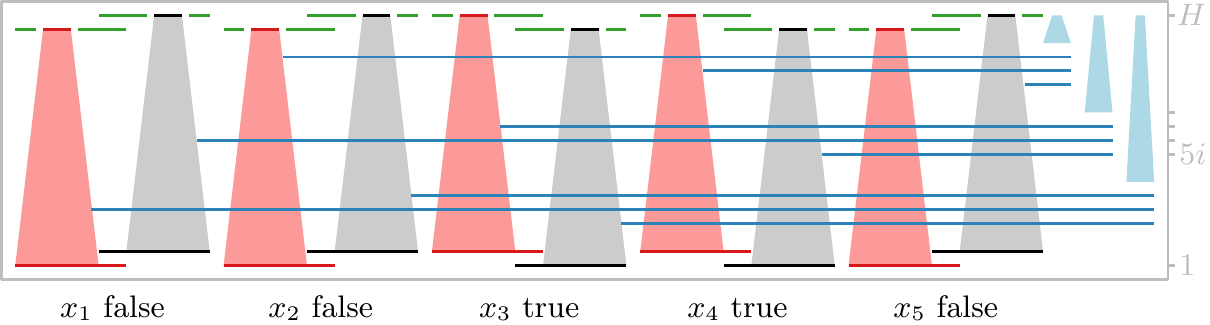}
		
		\bigskip
		
		\includegraphics[page=2]{containment-clause-gadget}
		\caption{Variable gadgets (red and gray) and clause gadgets (blue)
			for the 3-\textsc{Sat} instance
			$(\neg x_2 \vee \neg x_4 \vee x_5) \wedge (x_1 \vee \neg x_3 \vee x_4)
			\wedge (\neg x_1 \vee x_2 \vee x_3)$ with a fulfilling truth
			assignment (above) and a non-fulfilling assignment (below).  Note
			that the latter uses one color more (is one level higher).}
		\label{fig:clause-gadget}
	\end{figure}
	
	The key idea to transport a Boolean value from a variable gadget
	to a clause gadget is to add, for each occurrence of a literal
	$\ell_j \in \{x_j, \neg x_j\}$ in a clause~$C_i$, an ``arm''
	(blue intervals in \cref{fig:variable-gadget,fig:clause-gadget})
	that ends to the right of the clause gadget (the light blue
        Christmas tree) corresponding to~$C_i$ and starts
	immediately to the right of the $5i$-th interval of the gray tree
	(if $\ell_j=x_j$) or of the red tree (if $\ell_j=\neg x_j$)
	corresponding to~$\ell_j$.  The arm is represented by a sequence of
	intervals that are separated by a small gap within each
	Christmas tree of each clause gadget passed by the arm
	(such that, for any two arms, their gaps are disjoint
	and the resulting intervals do not contain each other).
	Assuming that the total number of colors is~$H$,
	two intervals of the same arm that are separated by a gap
	need to get the same color because, at the gap, $H-2$
	colors are occupied by other intervals and the one remaining
	``wrong'' color is blocked due to the green intervals of the
        variable gadgets; see \cref{fig:variable-gadget,fig:clause-gadget}.
	The green intervals are contained by the blue intervals of
        the arms and need to get color~$H$ or~$H-1$.
	
	If there is a satisfying truth assignment for $\varphi$, then
        there is a proper coloring that colors the variable gadgets
        such that they represent this truth assignment.  As for every
        clause~$C_i$, at least one of its literals in $C_i$ is true,
	the corresponding arm can use color~$5i$.
	Then, the arms of the other literals that occur in~$C_i$
	can use colors $5i+1$ and $5i+2$.  This allows the
	light blue Christmas tree representing clause~$C_i$,
	which has height $H - 5i - 2$ and is contained
	in the rightmost interval of each of these arms,
	to use the colors $\{5i+3,5i+4,\dots,H\}$.
	
	Now suppose for a contradiction that there is no satisfying
	truth assignment for $\varphi$, but that there is a
	proper coloring with $H$ colors.
	This coloring assigns a truth value to each variable gadget
	(depending on whether the bottommost interval of
	the red or the gray tree has color~1).
	Clearly, there is a clause $C_i$ in $\varphi$ that
	is not satisfied by this truth assignment.
	Hence, none of the arms connecting the clause gadget of~$C_i$
	with its three corresponding variable gadgets
	can use color~$5i$.  Hence
	they must use colors $5i+1$, $5i+2$, and $5i+3$ (or higher).
	This forces the (blue) Christmas tree representing
	clause~$C_i$ to use colors $\{5i+4,\dots,H,H+1\}$.
	
	Thus, a proper
	coloring with $H$ colors exists if and only if
	$\varphi$ is satisfiable.
\end{proof}

\section{Coloring Bidirectional Interval Graphs Is NP-Hard}
\label{sec:bidirectional-NP-hard}

In this section we show that it is NP-hard to color a bidirectional
interval graph with a given number of colors.
For a set \cI of intervals and a function $o$ that maps every
interval in~\cI to an orientation (left-going or right-going),
let $\cB[\cI,o]$ be the bidirectional interval
graph induced by~\cI and~$o$.

\begin{theorem}
  \label{thm:bidirectional-NP-hard}
  Given a set \cI of intervals with orientations $o$ and a positive
  integer $k$, it is NP-hard to decide whether $k$ colors suffice to
  color $\cB[\cI,o]$, that is, whether $\chi(\cB[\cI,o]) \le k$.
\end{theorem}

\begin{proof}
  We use the same ideas as in the proof of
  \cref{thm:containment-NP-hard}, but now
  we reduce from \textsc{Monotone 3-Sat},
  the version of \textsc{3-Sat} where every clause
  contains only negated or only unnegated variables as literals.
  For an overview, see
  \cref{fig:bidi-variable-gadget,fig:bidi-clause-gadget}.  Let
  $\varphi = C_1 \wedge C_2 \wedge \dots \wedge C_m$ be the given
  instance of \textsc{Monotone 3-Sat} with variables
  $\{x_1, x_2, \dots, x_n\}$.  As before, let $H = 5(m+1)$ be the
  number of colors sufficient for coloring a yes-instance.
  
  We now construct variable and clause gadgets by specifying
  a set~\cI of intervals with orientations~$o$.
  Our intervals have pairwise distinct endpoints.  Our main building
  structures are \emph{left- and right-going staircases}.
  A left-going staircase (gray in
  \cref{fig:bidi-variable-gadget,fig:bidi-clause-gadget}) is an
  ordered set of left-going intervals that share a common point
  and whose left and right endpoints are in the order of the set.
  A right-going staircase is symmetric (red in
  \cref{fig:bidi-variable-gadget,fig:bidi-clause-gadget}).
  Observe that staircases in (bi)directional interval graphs behave
  like Christmas trees in containment graphs: they form totally
  ordered cliques.
  The \emph{height} of a staircase is the number of its intervals.
  In \cref{fig:bidi-clause-gadget}, we draw staircases
  as parallelograms and we indicate left- and
  right-going intervals by arrow heads.
  
  \begin{figure}[tb]
  	\centering
  	\includegraphics[page=2]{bidi-variable-gadget}
  	\hfill
  	\includegraphics[page=1]{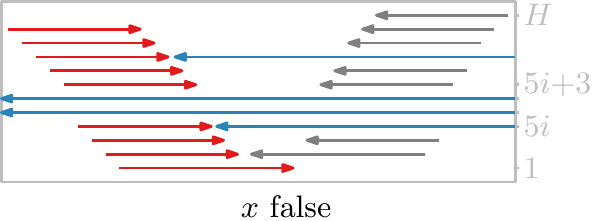}
  	
  	\caption{Variable gadget for the proof of
  		\cref{thm:bidirectional-NP-hard} in its two states.
  		Intervals directions are indicated by arrow heads.
  		The blue intervals extend (to the right)
  		to the clause gadgets of only negated variables.
  		The two blue arrow heads starting next to the red intervals
  		indicate that the clause~$C_i$ and a clause~$C_j$ with $j>i$
  		contain the literal $\neg x$.}
  	\label{fig:bidi-variable-gadget}
  \end{figure}
  
  \begin{figure}[tb]
  	\centering
  	\includegraphics[page=1,scale=.92]{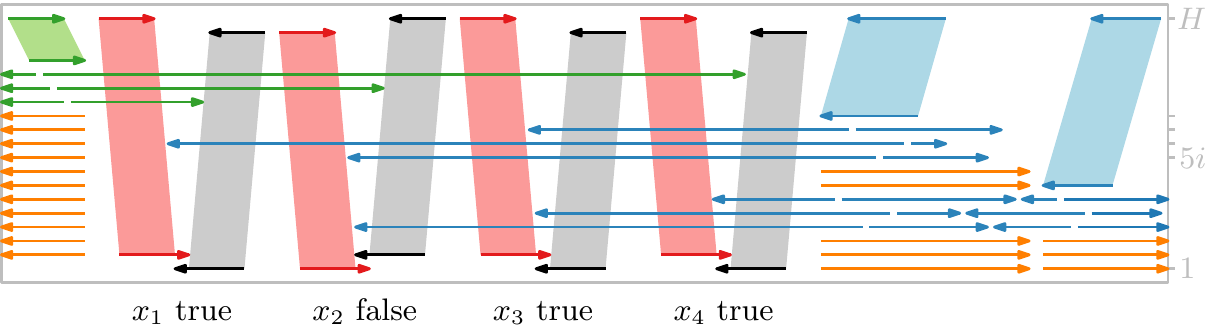}
  	
  	\medskip
  	
  	\includegraphics[page=2,scale=.92]{bidi-clause-gadget}
  	\caption{Variable gadgets (red and gray) and clause gadgets
                (light green and light blue)
  		for the \textsc{Monotone 3-Sat} instance
  		$(x_1 \vee x_2 \vee x_4) \wedge (\neg x_1 \vee \neg x_2 \vee \neg x_3)
  		\wedge (\neg x_2 \vee \neg x_3 \vee \neg x_4)$ with a fulfilling truth
  		assignment (top) and a non-fulfilling assignment (bottom),
  		which use $H$ and $H + 1$ colors, respectively.
  		Interval directions are indicated by arrow heads.}
  	\label{fig:bidi-clause-gadget}
  \end{figure}

  Let $i \in [m]$.  If the clause~$C_i$ has only negated literals,
  we again have three ``arms'' in~$C_i$ starting
  to the right of the $5i$-th interval of the red staircase
  of the three corresponding variables;
  see the blue intervals in \cref{fig:bidi-variable-gadget,fig:bidi-clause-gadget}.
  The intervals of these arms are left-going
  and they are not split by a gap at the other variable gadgets
  this time because now we need to contain the left-going
  intervals of the staircases to have edges instead of arcs
  in~$\cB[\cI,o]$.  The arms end below a left-going light blue staircase
  (see \cref{fig:bidi-clause-gadget} on the upper right side)
  of height $5(m-i)+3$ whose maximum color is~$H$
  if and only if none of the corresponding arms
  gets a color greater than $5i + 2$.
  
  Note that we need to avoid arcs between the arms.
  Therefore, we let every arm have a gap below
  each blue staircase that it passes.
  These gaps do not overlap and
  their order is inverse to the order of the left endpoints
  of the involved intervals.
  We continue the arm with a right-going interval
  (to avoid an arc with the blue staircase and the other arms)
  ending to the right of the blue staircase,
  where we continue again with a left-going interval.
  Consider such an arm with intervals $I$ and $I'$
  (going in different directions) around a gap.
  At this gap, it is important that
  no color smaller than the color of $I$ is available for $I'$,
  forcing $I'$ to also get the color of~$I$.
  Hence, we add long left-going intervals blocking
  every color not occupied by an arm or the blue staircase;
  see the orange intervals in \cref{fig:bidi-clause-gadget}.
  
  For every clause with only unnegated literals,
  we use the same construction but mirrored
  (connecting to the gray staircases);
  see the green intervals and light green staircases
  depicted in the top left corner of \cref{fig:bidi-clause-gadget}.
  
  We now have the same conditions as
  in the proof of \cref{thm:containment-NP-hard}:
  If there is a satisfying truth assignment for $\varphi$,
  there is a proper coloring, which colors
  the variable gadgets such that they represent this truth assignment.
  Again, for every $i \in [m]$, clause~$C_i$ contains at least one
  literal set to true.  Hence, the corresponding arm can get
  color~$5i$, and the other two arms and the (light blue or light
  green) staircase of~$C_i$ get colors $\{5i+1, 5i+2, \dots, H\}$.
  
  If there is no satisfying truth assignment for $\varphi$,
  then there is a clause~$C_i$ none of whose arms
  (as a whole) can get color~$5i$.
  If each arm occupies only one layer,
  then an interval of the clause gadget of~$C_i$ requires color $H+1$.
  If there is an arm occupying more than one layer,
  then below a clause gadget of some clause $C_{i'}$,
  there are two colors blocked by this arm (at one of its gaps).
  Then, however, an interval of the clique of size $H-1$ at this gap
  belonging to the (light blue or light green) staircase of~$C_{i'}$,
  to the other arms, or to the orange ``blocker'' intervals
  requires color~$H+1$; see \cref{fig:bidi-clause-gadget}
  for such an example.
\end{proof}

\section{Coloring General Mixed Interval Graphs}
\label{sec:coloring-general}

In this section we consider a further generalization of mixed interval
graphs.  We are dealing with an interval graph $G$ whose edges can be
arbitrarily oriented (or stay undirected).  In other words, the edge
directions are not related to the geometry of the intervals.

Observe that a proper coloring of $G$ exists if and only if $G$ does
not contain a directed cycle.  Let $\chi(G)$ denote the minimum number
of colors in a proper coloring of $G$, if it exists, or $\infty$ otherwise.
We point out that the existence of a directed cycle can be determined
in polynomial time (using, for example, depth-first search).

Note that clearly we have $\omega(G) \leq \chi(G)$.
However, there is another parameter that enforces a large chromatic
number even in sparse graphs.
A \emph{directed path} (of length $t$) in $G$ is a sequence of
vertices $\langle v_1,v_2,\dots,v_{t+1} \rangle$, such that, for each
$i \in [t]$, the arc $(v_i,v_{i+1})$ exists.
Let $\lambda(G)$ denote the length of a longest directed
path in~$G$.

Note that the vertices in a directed path receive pairwise distinct
colors in any proper coloring.  Thus we have $\chi(G) \ge \lambda(G)+1$,
and consequently $\chi(G) \ge \max \{\omega(G), \lambda(G)+1\}$.

\begin{theorem}
  \label{thm:mixed-upper-bound}
  Let $G$ be a mixed interval graph without directed cycles.
  Then $\chi(G) \le (\lambda(G)+1) \cdot \omega(G)$.
\end{theorem}
\begin{proof}
  Let $V$ denote the vertex set of $G$.  Let $G^{\to}$ be the graph
  obtained from $G$ by removing all edges.  Clearly,
  $G^{\to}$ is a DAG.  We partition $V$ into \emph{layers}
  $L_0,L_1,\ldots$ as follows.  The set~$L_0$ consists of the vertices
  that are sources in $G^{\to}$, i.e., they do not have incoming
  arcs.  Then, for $i=1,2,\ldots$, we iteratively define~$L_i$ to be
  the set of sources in $G^{\to} \setminus \bigcup_{j=0}^{i-1} L_j$.
  Note that $\lambda(G) = \max \{ i \colon L_i \ne \emptyset \}$.
  For $x \in V$, let $\ell(x) \in \{0,\dots,\lambda(G)\}$ denote the
  unique~$i$ such that $x \in L_i$.

  Recall that the underlying undirected graph of $G$,
  $U(G)$, is an (undirected) interval graph, and thus
  $\chi(U(G))=\omega(U(G))=\omega(G)$.  Let $c \colon V \to [\omega(U(G))]$
  be an optimal proper coloring of~$U(G)$.

  Now we define a coloring $f$ of $G$: for~$x \in V$, 
  let $f(x) = \ell(x) \cdot \omega(G)+c(x)$.
  Note that $1 \le f(x) \le (\lambda(G)+1) \cdot \omega(G)$.  We claim
  that $f$ is a proper coloring.

  Consider an edge $\{x,y\}$. As this is also an edge in $U(G)$,
  we obtain that $c(x) \neq c(y)$, and so $f(x) \neq f(y)$.  Now
  consider an arc $(x,y)$.  Its existence
  implies that $\ell(x) < \ell(y)$, and thus $f(x) < f(y)$.
\end{proof}

For some instances, the above bound is asymptotically tight.

\begin{proposition}
  \label{prop:mixed-lower-bound}
  There is an infinite family $(G_k)_{k \ge 1}$ of mixed interval
  graphs with $|V(G_k)|=2k^2$, $\lambda(G_k)=k-1$, $\omega(G_k)=2k$,
  and $\chi(G_k)=(k+1) \cdot k=(\lambda(G_k)+2)\cdot\omega(G_k)/2$.
\end{proposition}

\begin{proof}
  Let $\cI_k=\cI_{k,1} \cup \cI_{k,2} \cup \dots \cup \cI_{k,k}$ be a set of
  $k^2$ intervals defined as follows; see \cref{fig:mixed-example}
  for~$\cI_4$.  For $i \in [k]$, let $\cI_{k,i}$ be a multiset that
  contains $k$ copies of the interval $[6i,6i+8]$.  Similarly,
  let $\cI_k'=\cI_{k,1}' \cup \cI_{k,2}' \cup \dots \cup \cI_{k,k}'$
  be a set of $k^2$ intervals such that $\cI_k'$ is the image of
  mirroring $\cI_k$ at the point $x=6k+7$.  Note that, for
  $i \in [k-1]$, every interval in $\cI_{k,i}$ intersects every
  interval in~$\cI_{k,i+1}$ and every interval in $\cI_{k,i}'$
  intersects every interval in~$\cI_{k,i+1}'$.  Additionally, every
  interval in $\cI_{k,k}$ intersects every interval in $\cI_{k,k}'$.  

  Let $G_k$ be a mixed interval graph for the set~$\cI_k\cup\cI_k'$.
  We direct the edges of~$G_k$ as follows.  Let $\{I,I'\}$
  be a pair of intervals in $\cI_k \cup \cI_k'$ that intersect each other.
  If~$I$ and~$I'$ are copies of the same interval, then $\{I,I'\}$ is an
  edge of~$G_k$.  Otherwise, $(I,I')$ is an arc of~$G_k$ if $\{I,I'\}
  \subseteq \cI_k$ and~$I$ lies further to the left than~$I'$, if
  $\{I,I'\} \subseteq \cI_k'$ and~$I$ lies further to the right
  than~$I'$, or if $(I,I') \in \cI_{k,k} \times \cI_{k,k}'$. 
  It is easy to see that $G_k$ has the desired properties.
\end{proof}

Note that the mixed interval graphs that we constructed in the proof
above are even {\em directional} interval graphs.

\begin{figure}[tb]
  \centering
  \includegraphics{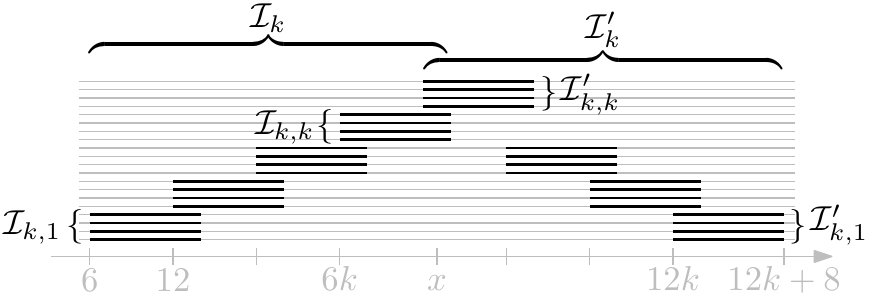}
  \caption{For any $k \ge 1$, the set $\cI_k\cup\cI_k'$ of
    intervals gives rise to a mixed interval graph~$G_k$ with
    $2k^2$ vertices, $\lambda(G_k)=k-1$, $\omega(G_k)=2k$, and
    $\chi(G_k)=(k+1) \cdot k=(\lambda(G_k)+2)\cdot\omega(G_k)/2$.}
  \label{fig:mixed-example}
\end{figure}

\section{Open Problems}

The obvious open problems are improvements to the results in
\cref{tab:results}, in particular: Is there a constant-factor
approximation algorithm for coloring general mixed interval graphs?
For applications in graph drawing, a better-than-2
approximation for coloring bidirectional interval graphs is of
particular interest.

Is there a linear-time recognition algorithm for directional or
containment interval graphs?
Is there a polynomial-time recognition algorithm for bidirectional
interval graphs?  

Using a reduction from \textsc{Max-3-SAT} instead of \textsc{3-SAT},
it may be possible to adjust our NP-hardness proofs in order to show
APX-hardness.  To this end, the difference in the number of colors
needed to color a yes-instance and the number of colors needed to
color a no-instance would have to be proportional to the number of
clauses that cannot be satisfied.  We were not able to force such a
large difference, hence we leave the APX-hardness of (or the existence
of a PTAS for) coloring containment and birectional interval graphs
open.

\subsubsection*{Acknowledgments.}

We are indebted to Krzysztof Fleszar, Zbigniew Lonc, Karolina Okrasa,
and Marta Piecyk for fruitful discussions.
Additionally, we acknowledge the welcoming and productive atmosphere
at the workshop Homonolo 2022, where some of the work was done.

\bibliography{mixed-interval-graphs}

\vfill\pagebreak\appendix\section*{Appendix}

\section{Recognition of Containment Interval Graphs}
\label{sec:recognition-containment}

Here we present a recognition algorithm for containment interval graphs.
Given a mixed graph~$G$, our algorithm decides whether~$G$ is a containment interval graph.
If it is, the algorithm additionally constructs a set~\cI of intervals representing~$G$, \ie, with $\cC[\cI]$ isomorphic with~$G$.
The algorithm works in three phases.
In the first phase, a \PQ-tree representing the interval representations of the underlying undirected graph~$U(G)$ of~$G$ is constructed.
In the second phase, the algorithm carefully selects a rotation of the \PQ-tree.
This corresponds to fixing the order in which the maximal cliques appear in the interval representation of~$U(G)$.
This almost fixes the interval representation.
In the third phase, the endpoints of the intervals are perturbed so that the edges and arcs of~$G$ are represented correctly.
We achieve this by reducing our problem to that of finding a two-dimensional realizer of a partial order.

The main result of this section is the following theorem.
\ThmRecognition*
\label{thm:recognition*}

In \cref{sec:recognition-containment-pq}, we introduce the necessary
machinery of \PQ-trees.  Details of the second and the third phase are
described in \cref{sec:recognition-containment-rotation,sec:recognition-containment-endpoints}, respectively.

\subsection{\MPQ-Trees}
\label{sec:recognition-containment-pq}
For a set of pairwise intersecting intervals on the real line, let the \emph{clique point} be the leftmost point on the real line that lies in all the intervals.
Given an interval representation of an interval graph~$G$, we get a linear order of the maximal cliques of~$G$ by their clique points from left to right.
Booth and Lueker~\cite{LuekerB79} showed that a graph~$G$ is an interval graph if and only if the maximal cliques of~$G$ admit a \emph{consecutive arrangement}, \ie, a linear order such that, for each vertex~$v$, all the maximal cliques containing~$v$ occur consecutively in the order.
They have also introduced a data structure called \PQ-tree that encodes all possible consecutive arrangements of~$G$.
We present our algorithm in terms of modified \PQ-trees (\MPQ-trees, for short) as described by Korte and M{\"o}hring~\cite{KorteM85,KorteM89}.

An \emph{\MPQ-tree}~$T$ of an interval graph~$G$ is a rooted, ordered tree with two types of nodes: \PQP-nodes and \PQQ-nodes, joined by links.
Each node can have any number of children and a set of at least two consecutive links joining a \PQQ-node~$x$ with some (but not all) of its children is called a \emph{segment} of~$x$.
Further, each vertex~$v$ of~$G$ is assigned either to a \PQP-node, or to a segment of some \PQQ-node.
Based on this assignment, we \emph{store}~$v$ in the links of~$T$.
If~$v$ is assigned to a \PQP-node~$x$, we store~$v$ in the link just above~$x$ in~$T$ (adding a dummy link above the root of~$T$).
If~$v$ is assigned to a segment of a \PQQ-node~$x$, we store~$v$ in each link of the segment.
For a link $\{x,y\}$, let~$S_{xy}$ denote the set of vertices stored in $\{x,y\}$.
We say that~$v$ is \emph{above} (\emph{below}, resp.) a node~$x$ if~$v$ is stored in any of the links on the upward path (in any of the links on some downward path, resp.) from~$x$ in~$T$.
We write $A^T_x$ ($B^T_x$, resp.) for the set of all vertices of~$G$ that are above (below, resp.) node~$x$.
Observe that every vertex assigned to a \PQP-node~$x$ is above~$x$, and every vertex assigned to a segment of a \PQQ-node~$x$ is below~$x$.

The \emph{frontier} of~$T$ is the sequence of the sets $A^T_x$, where~$x$ goes through all leaves in~$T$ in the order of~$T$.
Every node of~$T$ with at least two children is \emph{branching}.
Given a \PQQ-node~$x$ of an \MPQ-tree~$T$, there are two \emph{rotations of $x$}:
having the order of the children of $x$ as in the original tree~$T$,
and reversing the order of the children of~$x$.
For a \PQP-node~$x$ of an \MPQ-tree~$T$ with $k$ children, there are $k!$ rotations of~$x$, each obtained by a different permutation of the children of~$x$.
Every tree that is obtained from~$T$ by a sequence of rotations of nodes (\ie, obtained by arbitrarily permuting the order of the children of \PQP-nodes and reversing the orders of the children of some \PQQ-nodes) is a \emph{rotation} of~$T$.
The defining property of an \MPQ-tree~$T$ of a graph~$G$ is that each leaf~$x$ of~$T$ corresponds to a maximal clique $A^T_x$ of~$G$ and the frontiers of rotations of~$T$ correspond bijectively to the consecutive arrangements of~$G$.
Observe that any two vertices adjacent in~$G$ are stored in links that are connected by an upward path in~$T$.
We say that~$T$ \emph{agrees} with an interval representation \cI of~$G$ if the order of the maximal cliques of~$G$ given by their clique points in \cI from left to right is the same as in the frontier of~$T$.
We assume the following properties of an \MPQ-tree of~$T$ (see \cite{KorteM89}, Lemma~2.2):
\begin{itemize}
	\item For a \PQP-node~$x$ with children $y_1,\ldots,y_k$, for every $i=1,\ldots,k$,
	there is at least one vertex stored in link $\{x,y_i\}$ or below $y_i$, \ie, $S_{xy_i} \cup B^T_{y_i} \neq \emptyset$.
	\item For a \PQQ-node~$x$ with children $y_1,\ldots,y_k$, we have $k \geq 3$.
	Further, for $S_i=S_{xy_i}$, we have:
	\begin{itemize}
		\item $S_1 \cap S_k = \emptyset$,
		$B^T_{y_1} \neq \emptyset$, $B^T_{y_k} \neq \emptyset$,
		$S_{1} \subsetneq S_{2}$, $S_{k} \subsetneq S_{k-1}$,
		\item $(S_i \cap S_{i+1})\setminus S_1 \neq \emptyset$, $(S_{i-1} \cap S_{i})\setminus S_k \neq \emptyset$, for $i=2,\ldots,k-1$.
	\end{itemize}
  \item For every \PQP-node~$y$ being a child of a \PQP-node~$x$, we have that $x$ is branching.
  \item For every \PQP-node~$y$ being a child of a \PQP-node~$x$, there is at least one vertex assigned to $x$ or $y$.
\end{itemize}
As a consequence, we get that a node~$x$ is not branching, then it is either a leaf \PQP-node, or a \PQP-node with just one child which is a \PQQ-node.

\subsection{Rotating \MPQ-Tree}
\label{sec:recognition-containment-rotation}
\begin{lemma}\label{lem:recognition_rotation}
	There is an algorithm that, given a containment interval graph~$G$, constructs an \MPQ-tree~$T$ that agrees with some containment representation of~$G$.
	The running time of this algorithm is in $O(nm)$
	where $n$ is the number of vertices of~$G$
	and $m$ is the total number of edges and arcs of~$G$.
\end{lemma}

\begin{proof}
	Given a mixed graph~$G$, if~$G$ is a containment interval graph, then clearly~$U(G)$ is an interval graph and we can construct an \MPQ-tree~$T$ of~$U(G)$ in linear time using the algorithm by Korte and M{\"o}hring~\cite{KorteM89}.
	Further, we have that arcs induce a transitive directed acyclic graph, \ie, for every arcs $(u,v)$ and $(v,w)$, there is an arc $(u,w)$ and there is no directed cycle in $G$.
	
	We call a rotation of~$T$ \emph{correct} if it agrees with some containment representation of~$G$.
	As we assume~$G$ to be a containment interval graph, there is at least one correct rotation of~$T$, and our goal is to find some correct rotation of~$T$.
	Our algorithm decides the rotation of every node in~$T$, one-by-one, in any top-down order, \ie, from the root to the leafs.
	Thus, when the rotation of a node~$x$ is to be decided, the rotation of every node above~$x$ is already decided.
	Our algorithm keeps the invariant that before, and after, deciding the rotation of every single node, there is at least one correct rotation of~$T$ that agrees with the rotation of the already decided nodes.
	The invariant is trivially satisfied before the first rotation, and because it is satisfied after the last rotation, the algorithm constructs a correct rotation of~$T$.
	
	From now on, we focus on choosing a rotation of a single branching node~$x$.
	Let $y_1,\ldots,y_k$ denote the children of~$x$.
	We have $k \geq 2$ ($k \geq 3$ when $x$ is a \PQQ-node), and for each $i=1,\ldots,k$, let $B_i=S_{xy_i} \cup B^T_{y_i}$.
	We have $\bigcup_{i=1}^k B_i = B^T_x \neq \emptyset$.
	
  Let $\tilde T$ denote the (unknown) set of all correct rotations of $T$ that agree with the rotations of the already decided nodes.
	We have, by our invariant, that $\tilde T \neq \emptyset$,
	and our goal is to choose a rotation of $x$ that agrees with at least one rotation in $\tilde T$.
	We call any such rotation of $x$ a \emph{correct rotation of $x$}.
	
	For each vertex~$v$ above~$x$ in~$T$, it is already decided if~$v$ is above some node of~$T$ that is to the left (right, resp.) of~$x$ in all rotations in~$\tilde T$ (as this depends only on the rotation of the nodes above~$x$ in~$T$).
	If it is, then there is a maximum clique that: includes $v$, does not include any of the vertices below~$x$, and in the frontier of every $T' \in \tilde T$ is strictly to the left (right, resp.) of all maximum cliques containing vertices below $x$.
	This means that in every interval representation that agrees with some $T' \in \tilde T$, the interval representing~$v$ contains a clique-point that is strictly to the left (right, resp.) of all left endpoints (right endpoints, resp.) of intervals representing vertices below~$x$ in~$T$.
	We call such a vertex~$v$ \emph{left-$1$-bound} (\emph{right-$1$-bound}, resp.).
	Observe, that a vertex can be both left-$1$-bound and right-$1$-bound.
	If a vertex is neither left-$1$-, nor right-$1$-bound, we call it \emph{$1$-unbound}.
	For every $\ell \geq 2$, if $v$ is $(\ell-1)$-unbound and has an edge to an $(\ell-1)$-left-bound vertex~$u$, we call~$v$ $\ell$-right-bound.
  Similarly, if $v$ is $(\ell-1)$-unbound and has an edge to an $(\ell-1)$-right-bound vertex~$u$, we call~$v$ $\ell$-left-bound.
	If a vertex is $(\ell-1)$-unbound and neither left-$\ell$-, nor right-$\ell$-bound, we call it $\ell$-unbound.
	Lastly, a vertex is \emph{$\ell$-bound} if it is left-$\ell$-bound or right-$\ell$-bound, \emph{bound} if it is $\ell$-bound for some $\ell\geq1$, and \emph{unbound} if it is $\ell$-unbound for every $\ell \geq 1$.
	
	Observe that the properties of \MPQ-trees guarantee that for a $1$-unbound vertex~$v$, we have that either~$x$ is a \PQP-node and~$v$ is assigned to~$x$, or~$x$ is a \PQQ-node with a parent \PQP-node~$z$, $x$ is the only child of $z$, and~$v$ is assigned to~$z$.

	Observe that, intuitively, it is ``natural'' for intervals representing vertices below a node to be contained in intervals representing vertices above it.
	Let $x$ be a node in $T$, $v \in A^T_x$, and $w \in B^T_x$.
	First notice that it is impossible to have an arc $(w,v)$, \ie, it is impossible for the interval of $w$ to contain the interval of $v$.
	As $x$ is a branching node, $w$ omits clique-points in the subtree of at least one child of $x$.
	These clique-points are contained in the interval of $v$.
	It is still possible to have an edge joining a vertex above $x$ and a vertex below $x$.
	Each such edge allows us to deduce some information on the correct rotations of $x$ and these edges are crucial to our algorithm.
	
\begin{claim}\label{clm:bound-pair}
  For every $\ell\geq1$,
	for an $\ell$-bound vertex $v$ and an $\ell$-unbound vertex $w$,
  there is no arc $(w,v)$.
\end{claim}
\begin{claimproof}
	We prove this by induction on $\ell$.
	For $\ell=1$, every interval of a $1$-bound vertex contains a clique-point such that if any other interval contains it, the corresponding vertex is also $1$-bound.
	For $\ell>1$, let $u$ be an $(\ell-1)$-bound vertex that certifies that $v$ is $\ell$-bound.
  By the induction hypothesis, there is no arc $(w,u)$, and there is no edge $\{u,w\}$ as $w$ is $\ell$-unbound.
  Thus, there is an arc $(u,w)$.
  Assuming to the contrary the existence of an arc $(w,v)$, we get an arc $(u,v)$ from transitivity of arcs.
  This contradicts with the existence of the edge $\{u,v\}$.
\end{claimproof}
	
\begin{claim}\label{clm:bound-path}
  First, let $v_1,\ldots,v_\ell$ be a sequence of bound vertices such
  that (i)~for every odd~$i \in [\ell]$, $v_i$ is left-$i$-bound;
  (ii)~for every even~$i \in [\ell]$, $v_i$ is right-$i$-bound; and
  (iii)~for every $i \in [\ell-1]$, there is an edge $\{v_i,v_{i+1}\}$.
  For every interval representation that agrees with some $T' \in \tilde T$, for every odd (even, resp.) $2 \leq i \leq \ell$,
	the interval representing~$v_{i}$ contains the left (right, resp.) endpoint of the interval representing~$v_{i-1}$.\\

  Second, let $v_1,\dots,v_\ell$ be a sequence of bound vertices such
  that (i)~for every odd~$i \in [\ell]$, $v_i$ is right-$i$-bound;
  (ii)~for every even~$i \in [\ell]$, $v_i$ is left-$i$-bound; and
  (iii)~for every $i \in [\ell-1]$, there is an edge $\{v_i,v_{i+1}\}$.
	For every interval representation that agrees with some $T' \in \tilde T$, for every odd (even, resp.) $2 \leq i \leq \ell$,
	the interval representing~$v_{i}$ contains the right (left, resp.) endpoint of the interval representing~$v_{i-1}$.
\end{claim}
\begin{claimproof}
	By the previous claim, we have that there is an arc $(v_i,v_j)$ for every
	$i \in [\ell -2]$ and $j \in \{i+2, \dots, \ell\}$.
  As the second part of the claim is symmetric, we prove only the first part.
	The proof is by induction on $i$.
	For $i=2$, we have: $v_1$ is left-$1$-bound, while~$v_2$ is not,
	and there is an edge $\{v_1,v_2\}$, which means that~$v_2$ contains the right endpoint of~$v_1$.
	For odd $i>2$, we have: $v_i$ is contained in~$v_{i-2}$, $v_{i-1}$ contains the right endpoint of~$v_{i-2}$ by induction, $v_{i}$ does not include the right endpoint of~$v_{i-1}$, and the edge $\{v_{i-1},v_{i}\}$ means that $v_{i}$ contains the left endpoint of~$v_{i-1}$.
	For even $i>2$, the argument is symmetric.
\end{claimproof}
		
	\subparagraph*{Rotating \PQQ-nodes.}
	Observe that each vertex $w \in B^T_x$ is present in at most one of $B_1$ or~$B_k$.
	(Recall that $B_i=S_{xy_i} \cup B^T_{y_i}$
	for $i \in [k]$ and $y_1, \dots, y_k$ are the children of~$x$.)
	
	We shall prove that if there is at least one edge joining a vertex $w$ below $x$ and a bound vertex $v$ above $x$, then there is only one correct rotation of $x$.
	If $v$ is left-$\ell$-bound (right-$\ell$-bound, resp.) then $x$ needs to be rotated so that $w$ is in the last (first, resp.) child of~$x$.
	Otherwise, when bound vertices above $x$ have arcs towards vertices below $x$, then both rotations are correct.

  Assume that there is an edge for some $\ell$-bound $v \in A^T_x$ and $w \in B^T_x$, and $\ell$ is minimum possible (\ie there are no edges joining $\ell'$-bound vertices and $B^T_x$ for $\ell' < \ell$).
  Assume $v$ is left-$\ell$-bound, as the other case is symmetric.
  We prove that $w$ must be assigned to or below the last child of~$x$ in every $T' \in \tilde T$.
	For $\ell=1$, $v$ is left-$1$-bound, and the left endpoint of $v$ is to the left of the left endpoint of $w$ in every containment representation that agrees with some $T' \in \tilde T$.
	Thus, to realize the edge $\{v,w\}$, the right endpoint of $v$ is to the left of the right endpoint of $w$, which requires $w$ to be in the last child of $x$ (as otherwise there is some clique-point to the right of $w$ that is in $v$).
	For $\ell>1$, let $u$ be the right-$(\ell-1)$-bound vertex with edge $\{u,v\}$.
  By Claim~\ref{clm:bound-path}, we know that $v$ contains the left endpoint of $u$.
	Because $\ell$ is minimum, $u$ and $w$ are not connected by an edge, and by Claim~\ref{clm:bound-path} the interval of $u$ contains the interval of $w$.
	Thus, $v$ contains the left enpoint of $w$, and in order to have the right endpoint of $w$ after the right endpoint of $v$,
	we need $w$ to be assigned to or below the last child of~$x$.
	
	Similarly if $v$ is right-bound, then $w$ must be in the first child in every $T' \in \tilde T$.

	For the second part, we assume that there is an arc from every bound vertex above $x$ towards every vertex below $x$.
	We know that there is also an arc from every bound vertex to every unbound vertex.
	Let $B$ denote the set of all vertices that are either below $x$ or unbound.
	Observe that any containment representation of $G$ has all the endpoints of intervals representing vertices in $B$ strictly inside the intersection of intervals representing bound vertices.
	Thus, reversing the order of all endpoints of intervals representing vertices in $B$ gives another containment representation of $G$.
	This other representation has the order of clique-points represented by the subtree of $x$ reversed.
	Thus, both orientations of $x$ are correct and algorithm can choose any of them.
	
	\subparagraph*{Rotating \PQP-nodes.}
	We are to choose the order of the children $y_1,\ldots,y_k$ of~$x$.
	Observe that in this case, for a \PQP-node, the sets $B_1, \dots, B_k$ are pairwise disjoint, and there are neither arcs, nor edges joining two different sets $B_i$ and $B_j$.

  Now, assume $k\geq3$, and observe that for every vertex $v \in A^T_x$ above $x$, and every vertex $w \in B_i$ assigned to or below a middle child of $x$, we have an arc $(v,w)$.
	This is because there is at least one clique-point below the first, and below the last child of $x$.
	Children assigned to or below middle children are not in these cliques.
	Thus, an interval representing $v$ must contain an interval representing $w$.

  Now, we call a child $y_i$ of $x$ to be \emph{special}, if there is an edge joining a vertex $v \in A^T_x$ with a vertex $w \in B_i$.
  We already know that there are at most two special children of~$x$, as otherwise $\tilde T = \emptyset$.
	Observe that if $\sigma$ is a correct rotation of $x$, then any $\sigma'$ that is obtained from $\sigma$ by arbitrarily permuting the middle children is also a correct rotation of $x$.
	This is because in the rotation $\sigma$, we have that all middle children of $x$ are not special, and there are neither edges nor arcs between different sets $B_i$ and $B_j$.

	Now, let us fix a single permutation $\psi$ of the children of $x$ in which every special child of $x$ is either the first, or the last child.
	Let $\psi'$ denote the permutation obtained by reversing $\psi$.
	It is easy to see that if there is a correct rotation of $x$ at all, then also either $\psi$ or $\psi'$ is a correct rotation of $x$.
	Now, we can apply the same reasoning as for the \PQQ-nodes.
	If there is at least one edge joining a vertex $w$ below $x$ and a bound vertex $v$ above $x$, then only one of $\psi$, or $\psi'$ is a correct rotation of $x$.
	Otherwise, both $\psi$, and $\psi'$ are correct rotations of $x$.

  \begin{claim}
    The rotation of a single node can be decided in time $O(n+m)$.
  \end{claim}
\begin{claimproof}
  For a \PQQ-node we first need to decide which vertices are left/right-$\ell$-bound for different $\ell$.
  We can first calculate the set $A^T_x$.
  Then traverse the tree upwards and in each node mark left/right-$1$-bound vertices.
  Then use BFS to decide which vertices in $A^T_x$ are left/right-$\ell$-bound for different $\ell$.
  This can be easily done in $O(n+m)$ time.
  
  Then, for each edge that connects an $\ell$-bound vertex $v$ with a vertex $w$, we need to decide if $w \in B_1$ or $w \in B_2$.
  Observe that the queries ``whether a vertex $w$ is in $B_i$?'' can be answered in constant time (by looking at the index of the first/last clique in the frontier that includes $w$, and on the index of the first/last clique in the frontier that is below a node $y_i$).
  Thus, we can decide the rotation of a \PQQ-node in $O(n+m)$ time.

  For a \PQP-node, we first need to decide which children are special.
  For this we need to calculate the set $A^T_x$, and sets $B^i$, and then for each edge check if it makes some child special.
  This can be easily done in $O(n+m)$ time.
  The rest of the analysis is the same as for a \PQQ-node.
\end{claimproof}

  As there are $n$ nodes to rotate, and by the previous claim, we conclude that the running time of the algorithm is $O(nm)$.
\end{proof}

\subsection{Perturbing Endpoints}
\label{sec:recognition-containment-endpoints}

\begin{lemma}\label{lem:recognition_endpoints}
	There is an algorithm that, given an \MPQ-tree~$T$ that agrees with
	some containment representation of a mixed graph~$G$, constructs a
	containment representation \cI of~$G$ such that $T$ agrees
	with \cI.
	The running time of this algorithm is in $O(nm)$
	where $n$ is the number of vertices of~$G$
	and $m$ is the total number of edges and arcs of~$G$.
\end{lemma}
\begin{proof}
  The frontier of~$T$ fixes the left-to-right order of clique-points of maximal-cliques in~$G$.
  We need to respect that order, but still we have some freedom in choosing the exact locations of the endpoints.
  For any vertex~$v$, let $L_v$ ($R_v$, resp.) denote the index of the first (last, resp.) clique in the frontier of~$T$ that includes $v$.
  For any $L \ge 1$ ($R \ge 1$, resp.), the \emph{left-$L$-group} (\emph{right-$R$-group}, resp.) is the set of all vertices $v$ with $L_v=L$ ($R_v=R$, resp.).
  It is easy to see that any interval representation of~$U(G)$ that agrees with~$T$ can be stretched so that, for every vertex $v$, we have that
  the left endpoint $l_v$ of~$v$ is a real in the open interval $l_v \in (L_v-\frac{1}{2},L_v)$, and the right endpoint $r_v$ of~$v$ is in the interval $r_v \in (R_v,R_v+\frac{1}{2})$.
  Obviously, any representation satisfying these conditions on the locations of the endpoints agrees with~$T$.
  We are free to determine the order among endpoints in each group independently, so that the resulting intervals are a containment representation of~$G$.

  We will now collect different order constraints on the relative location of pairs of the endpoints.
	First, consider two adjacent vertices $u$ and $v$ with $L_u = L_v$, and $R_u < R_v$.
	As the right endpoints are in different right-groups, we have $r_u < r_v$.
	If there is the edge $\{u, v\}$, then, regarding the relative order, of the left endpoints, we need to have $l_u<l_v$.
  If there is the arc $(u, v)$, then we need to have $l_u>l_v$.
	The arc $(v, u)$ is impossible to realize.
	Similarly, if any of the left/right-groups is common for $u$ and $v$, but the other one is different, the relative order of endpoints is fixed.

  We have collected information about all pairs of vertices $u$ and $v$, except for these with $L_u = L_v$ and $R_u = R_v$.
  In this case,	if there is an arc $(u, v)$ or $(v, u)$, then again the relative order of left and right endpoints is fixed.
	
  Now, we assume that there is an edge $\{u,v\}$ and we want to say something about the relative order of the endpoints.
  Clearly, we have $l_u < l_v \iff r_u < r_v$, but we would like to decide the correct order.
  For a third vertex $w$, we say that $w$ \emph{behaves the same} (\emph{differently}, resp.) on $u$ and $v$, when $w$ is connected to $u$ with the same type of connection (edge, arc towards, arc from) as to $v$ (otherwise, resp.).
	Assume that there is a vertex~$w$ with $L_u = L_v = L_w$ and $R_u = R_v \ne R_w$ that behaves differently on $u$ and $v$.
  Note that we do not have two arcs in different directions joining $w$ with $u$ and $v$ as it would imply an arc between $u$ and $v$ by transitivity.
	Thus, we assume w.l.o.g.\ that $w$ is connected by an arc $(u, w)$ or $(w, u)$ with $u$, but with an edge $\{v, w\}$ with $v$.
	Then, the relative order of $u$ and $v$ is fixed because there is only one relative order of the three left endpoints that allows for this situation.
  Similarly for the case $L_u = L_v \ne L_w$ or $R_u = R_v = R_w$.

  For a pair of vertices $u$, $v$ with the edge $\{u,v\}$, $L_u=L_v$, and $R_u=R_v$,
  if the above rule gives us the relative order of their endpoints,
  we call such pair \emph{decided}.
  Otherwise, it is \emph{undecided}.
	While there are undecided pairs, we propagate the order of the decided pairs as follows.
	Consider a vertex~$w$ with $L_u = L_w$ and $R_u = R_w$
	that behaves differently on~$u$ and~$v$,
	\ie, there is an arc $(u, w)$ (or $(w, u)$) and an edge $\{v, w\}$
	(two arcs in different directions is not possible as argued before),
	and let $\{v, w\}$ be decided.
	Then, the relative orders of the endpoints of $v$ and $w$
	and the ones of $u$ and $w$ are fixed, which means
	that the relative orders of the endpoints of $u$ and $v$ follow.
	From now on the pair $u$, $v$ is also decided and we apply this procedure as long as it is possible.

  At this point, if there are some undecided pairs left, choose any vertex $u$ from an undecided pair, and let $U$ be the set of vertices reachable from $u$ by a path of undecided edges.
  We have $|U|\geq2$, and every vertex $w \notin U$ behaves the same on any two vertices in $U$
  as otherwise we would have applied
  one of the order constraints described above.
  We remove all vertices in $U$ except $u$ and solve the smaller instance of the problem.
  
  We can find all order constraints in $O(nm)$ time
  as it suffices to consider each edge together with each vertex.

  Now, it remains to prove, that we can insert back the vetices in $U$ to the solution.
  Observe that each vertex $v$ in $U$ can be placed in the position of $u$ and this position satisfies all order constraints of $u$ against vertices not in $U$.
  Thus, we will put all the left (right, resp.) endpoints of vertices in $U$ in an $\varepsilon$ range around the left (right, resp.) endpoint of $u$.
  Consider the mixed graph induced by the vertices in $U$.
  This is a complete mixed acyclic graph and can be seen as a partially ordered set.
  
  A \emph{partially ordered set}, or a \emph{poset} for short, is a transitive directed acyclic graph.
  A poset~$P$ is \emph{total} if, for every pair of vertices~$u$ and~$v$, there is either an arc $(u,v)$ or an arc $(v,u)$ in~$P$.
  We can conveniently represent a total poset~$P$ by a linear order of its vertices $v_1 < v_2 < \dots < v_n$ meaning that there is an arc $(v_i,v_j)$ for each $1 \leq i < j \leq n$.
  A poset~$P$ is \emph{two-dimensional} if the arc set of~$P$ is the intersection of the arc sets of two total posets on the same set of vertices as~$P$.
  McConnell and Spinrad~\cite{McConnellS99} gave a linear-time algorithm that, given a directed graph~$D$ as input, decides whether~$D$ is a two-dimensional poset.
  If the answer is yes, the algorithm also constructs a \emph{realizer}, that is, (in this case) two linear orders $(R_1,R_2)$ on the vertex set of~$D$ such that
  \begin{align*}
  \textrm{arc $(u,v)$ is in~$D$} & \iff
  \left[(\textrm{$u < v$ in $R_1$}) \wedge (\textrm{$u < v$ in $R_2$})\right] \textrm{.}
  \end{align*}

  \begin{claim}
    The mixed graph induced by $U$ is a containment interval graph if and only if the poset of $U$ is two-dimensional.
  \end{claim}
  \begin{claimproof}
    First, for the ``if'' direction, assume that $U$ is two-dimensional, and let $L_1$, $L_2$ be two linear extensions of $U$ such that $U$ is the intersection of $L_1$ and $L_2$.
    Now, we construct the containment interval representation of $G[U]$ in the following way:
    We choose the locations of the left enpoints of the intervals representing the vertices in $U$ in the open interval $(-\frac{1}{2},0)$ so that their left-to-right order is exactly as in $L_1$.
    Similarly, we place right endpoints in the interval $(0, \frac{1}{2})$ so that their right-to-left order is exactly as in $L_2$.
    Now, for an arc $(u,v)$ we have that $u \leq v$ in the poset, and $u \leq_{L_1} v$ and $u \leq_{L_2} v$.
    Thus, the left endpoint of $u$ is to the left of left endpoint of $v$, and right endpoint of $u$ is to the right of the right endpoint of $v$, as required.
    Conversely, for an edge $\{u,v\}$ we have that $u$ and $v$ are incomparable in the poset, and $u \leq_{L_1} v$ and $v \leq_{L_2} u$ (or the other way around, both inequalities are reversed).
    Thus, the resulting intervals overlap without containment, and the resulting set of intervals is a containment representation of~$G[U]$.

    For the other direction, given any containment interval representation, let $L_1$ be a linear order on $U$ given by the left-to-right order of the left endpoints of the intervals, and $L_2$ be given by the right-to-lleft order of the right endpoints.
    Now, after the previous argument, it is easy to see that $L_1, L_2$ is a realizer of the poset of $U$, and $U$ is of two-dimensional.
  \end{claimproof}

  By the claim, we can realize $G[U]$ within small ranges designated to the endpoints of $u$.
  As the running time of this step is linear,
  the resulting running time of this algorithm is in $O(nm)$.
\end{proof}

\subsection{Final Proof}

\cref{thm:recognition} follows easily from \cref{lem:recognition_rotation,lem:recognition_endpoints}.

\begin{proof}[Proof of \cref{thm:recognition}]
	Our algorithm, given a containment interval graph~$G$, applies the
	algorithm from Lemma~\ref{lem:recognition_rotation} to obtain an \MPQ-tree~$T$ that agrees with some containment representation of~$G$.
	Then, using Lemma~\ref{lem:recognition_endpoints}, it constructs a containment representation of~$G$.
	If any of the phases fails, then we know that~$G$ is not a containment interval graph, and we can reject the input.
	Otherwise, our algorithm accepts the input and returns a containment representation of~$G$.
  As both phases can be implemented to run in $O(nm)$ time, we get that our algorithm recognizes containment interval graphs in $O(nm)$ time.
\end{proof}

\end{document}